\documentclass[11pt]{article}

\usepackage{cite}
\usepackage{booktabs}
\usepackage{hyperref}
\usepackage{graphicx}
\usepackage{subfig}
\usepackage{paralist}

\usepackage{geometry}
\usepackage{amsthm}
\usepackage{amsmath}
\usepackage{amssymb}

\allowdisplaybreaks[1]

\newtheorem{theorem}{Theorem}
\newtheorem{lemma}[theorem]{Lemma}

\newtheorem{observation}[theorem]{Observation}

\makeatletter
\def\@endtheorem{\endtrivlist}
\makeatother

\newcommand{\select}[3]{\mathrm{select}_{#1}(#2,#3)}

\newcommand{\depth}[1]{\mathrm{depth}(#1)}
\newcommand{\degree}[1]{\mathrm{degree}(#1)}
\newcommand{\parent}[1]{\mathrm{parent}(#1)}
\newcommand{\levelancestor}[2]{\mathrm{level\_ancestor}(#1,#2)}
\newcommand{\childrank}[1]{\mathrm{child\_rank}(#1)}
\newcommand{\childselect}[2]{\mathrm{child\_select}(#1,#2)}
\newcommand{\numdescendants}[1]{\mathrm{num\_descendants}(#1)}
\newcommand{\preorderrank}[1]{\mathrm{pre\_rank}(#1)}
\newcommand{\preorderselect}[1]{\mathrm{pre\_select}(#1)}
\newcommand{\height}[1]{\mathrm{height}(#1)}
\newcommand{\lca}[2]{\mathrm{lca}(#1,#2)}

\newcommand{\preorderselectb}[2]{\mathrm{pre\_select}_{#1}(#2)}

\newcommand{\substr}[3]{#1[#2..#3]}
\newcommand{\subtree}[2]{#1\langle#2\rangle}

\newcommand{\entropy}[1]{H^*(#1)}
\newcommand{\entropyr}[1]{E(#1)}
\newcommand{\er}[1]{e(#1)}

\newcommand{\D}[2]{\mathcal{D}_{#1,#2}} % tree decomposition of a tree
\newcommand{\TL}[2]{\mathcal{T}_{#1,#2}} % reduced tree

\newcommand{\func}{f}
\newcommand{\funcp}{f'}
\newcommand{\gfunc}{g}
\newcommand{\A}{\mathcal{A}}
\newcommand{\F}{\mathcal{F}}
\newcommand{\Num}[2]{\mathcal{N}(#1,#2)}

\newcommand{\fsize}[1]{|#1|}
\newcommand{\fsizeright}[1]{s_{#1}}
\newcommand{\fsizeboundary}[1]{s'_{#1}}
\newcommand{\fsizemax}[1]{s''_{#1}}
\newcommand{\fdegree}[1]{d_{#1}}
\newcommand{\fdistance}[1]{l_{#1}}
\newcommand{\fpreorder}[1]{p_{#1}}

\newcommand{\partitionb}{\mathrm{Decompose}}
\newcommand{\joinb}{\mathrm{Join}}
\newcommand{\partition}[1]{\partitionb(#1)}
\newcommand{\join}[1]{\joinb(#1)}

\newcommand{\packb}{\mathrm{pack}}
\newcommand{\pack}[2]{\packb(#1,#2)}

\begin{document}

\title{Representation of ordered trees with a given degree distribution}
\author{Dekel Tsur%
\thanks{Department of Computer Science, Ben-Gurion University of the Negev.
Email: \texttt{dekelts@cs.bgu.ac.il}}}
\date{}
\maketitle

\begin{abstract}
%The tree degree entropy of an ordered tree $T$ is
%$\entropy{T} = \frac{1}{n}\sum_{i\colon n_i>0} n_i \log \frac{n}{n_i}$,
%where $n$ is the number of nodes in $T$ and $n_i$ is the number of nodes in $T$
%with $i$ children.
The degree distribution of an ordered tree $T$ with $n$ nodes
is $\vec{n} = (n_0,\ldots,n_{n-1})$,
where $n_i$ is the number of nodes in $T$ with $i$ children.
Let $\mathcal{N}(\vec{n})$ be the number of trees with degree distribution
$\vec{n}$.
We give a data structure that stores an ordered tree $T$ with $n$ nodes
and degree distribution $\vec{n}$ using
$\log \mathcal{N}(\vec{n})+O(n/\log^t n)$ bits for every constant $t$.
The data structure answers tree queries in constant time.
This improves the current data structures with lowest space for ordered trees:
The structure of Jansson et al.\ [JCSS 2012] that uses
$\log\mathcal{N}(\vec{n})+O(n\log\log n/\log n)$ bits,
and the structure of Navarro and Sadakane [TALG 2014] that uses
$2n+O(n/\log^t n)$ bits for every constant $t$.

%\noindent Keywords: succinct data-structures; ordered trees.
\end{abstract}

\section{Introduction}

A problem which was extensively studied in recent years is designing a
succinct data structure that stores a tree while supporting queries on the tree,
like finding the parent of a node, or computing the lowest common ancestor
of two
%nodes
%This problem has been studied both for static trees
%and dynamic trees~\cite{
nodes~\cite{Jacobson89,MunroR01,BenoitDMRRR05,DelprattRR06,GearyRR06,GearyRRR06,GolynskiGGRR07,RamanRS07,HeMS12,JanssonSS12,MunroRRR12,FarzanM14,NavarroS14,
MunroRS01,RamanR03,FarzanM11,ArroyueloDS16,GuptaHSV07}.
The problem of storing a static ordinal tree was studied
in~\cite{Jacobson89,MunroR01,BenoitDMRRR05,DelprattRR06,GearyRR06,GearyRRR06,
GolynskiGGRR07,JanssonSS12,MunroRRR12,FarzanM14,NavarroS14}.
These paper show that an ordinal tree with $n$ nodes can be stored
using $2n+o(n)$ bits while answering queries in constant time.
The space of $2n+o(n)$ bits matches the lower bound of $2n-\Theta(\log n)$ bits
for this problem.
In most of these papers, the $o(n)$ term is $\Omega(n\log\log n/\log n)$.
The only exception is the data structure of
Navarro and Sadakane~\cite{NavarroS14} which uses $2n+O(n/\log^t n)$ bits for
every constant $t$.

%The tree degree entropy of an ordered tree $T$ is
%$\entropy{T} = \frac{1}{n}\sum_{i\colon n_i>0} n_i \log \frac{n}{n_i}$,
%where $n$ is the number of nodes in $T$ and $n_i$ is the number of nodes in $T$
%with $i$ children.
Jansson et al.~\cite{JanssonSS12} studied the problem of storing a tree
with a given degree distribution.
The \emph{degree distribution} of an ordered tree $T$ with $n$ nodes
is $\vec{n} = (n_0,\ldots,n_{n-1})$,
where $n_i$ is the number of nodes in $T$ with $i$ children.
Let $\mathcal{N}(\vec{n})$ be the number of trees with degree distribution
$\vec{n}$.
%introduced an entropy function on trees,
%called \emph{tree degree entropy}.
%The tree degree entropy of a tree $T$ is
%$\entropy{T} = \frac{1}{n}\sum_{i\colon n_i>0} n_i \log \frac{n}{n_i}$,
%where $n$ is the number of nodes in $T$ and $n_i$ is the number of nodes in $T$
%with $i$ children.
Jansson et al.\ showed a data structure that stores a tree $T$ with degree
distribution $\vec{n}$
using $\log\mathcal{N}(\vec{n})+O(n\log\log n/\log n)$ bits,
and answers tree queries in constant time.
This data structure is based on Huffman code that stores the sequence of
node degrees (according to preorder).
A different data structure was given by Farzan and Munro~\cite{FarzanM14}.
The space complexity of this structure is
$\log\mathcal{N}(\vec{n})+O(n\log\log n/\sqrt{\log n})$ bits.
The data structure of Farzan and Munro is based on a tree decomposition
approach.

In this paper, we give a data structure that stores a tree $T$ using
$\log\mathcal{N}(\vec{n})+O(n/\log^t n)$ bits, for every constant $t$.
%The data structure answers tree queries in constant time.
This results improve both the data structure of
Navarro and Sadakane~\cite{NavarroS14}
(since $\log \mathcal{N}(\vec{n})\leq 2n$)
and the data structure of Jansson et al.~\cite{JanssonSS12}.
Our data structure supports many tree queries which are answered
in constant time.
See Table~\ref{tab:queries} for some of the queries supported by our data
structure.

\begin{table}
\caption{Some of the tree queries supported by the our data structure.
\label{tab:queries}}
\medskip
\centering
\begin{tabular}{lp{10.5cm}}
\toprule
Query & Description \\
\midrule
$\depth{x}$ & The depth of $x$. \\
$\height{x}$ & The height of $x$. \\
$\numdescendants{x}$ & The number of descendants of $x$. \\
$\parent{x}$ & The parent of $x$. \\
$\lca{x}{y}$ & The lowest common ancestor of $x$ and $y$. \\
$\levelancestor{x}{i}$ & The ancestor $y$ of $x$ for which
						$\depth{y} = \depth{x}-i$. \\
$\degree{x}$ & The number of children of $x$. \\
$\childrank{x}$ & The rank of $x$ among its siblings. \\
$\childselect{x}{i}$ & The $i$-th child of $x$. \\
$\preorderrank{x}$ & The preorder rank of $x$. \\
$\preorderselect{i}$ & The $i$-th node in the preorder. \\
\bottomrule
\end{tabular}
\end{table}

Our data structure is based on two components. The first component is the
tree decomposition method of Farzan and Munro~\cite{FarzanM14}.
While Farzan and Munro used two levels of decomposition, we use an
arbitrarily large constant number of levels.
The second component is the aB-tree of
Patrascu~\cite{Patrascu08}, which is a structure for storing an array of
poly-logarithmic size with almost optimal space,
while supporting queries on the array in constant time.
This structure has been used for storing trees
in~\cite{NavarroS14,Tsur_labeled}.
However, in these papers the tree is converted to an array and tree queries are
handled using queries on the array.
In this paper we give a generalized aB-tree which can directly store an object
from some decomposable family of objects.
This generalization may be useful for the design of succinct data structures
for other problems.

The rest of this paper is organized as follows.
In Section~\ref{sec:tree-decomposition} we describe the tree decomposition
of Farzan and Munro~\cite{FarzanM14}.
In Section~\ref{sec:aB-tree} we generalize the aB-tree structure of
Patrascu~\cite{Patrascu08}.
Finally, in Section~\ref{sec:data-structure}, we describe our data structure
for ordinal trees.

%\section{Preliminaries}

\section{Tree decomposition}\label{sec:tree-decomposition}
One component of our data structure is the tree decomposition of
Farzan and Munro~\cite{FarzanM14}.
In this section we describe a slightly modified version of this decomposition.
\begin{lemma}\label{lem:tree-decomposition}
For a tree $T$ with $n$ nodes and an integer $L$, there is a collection
$\D{T}{L}$ of subtrees of $T$ with
the following properties.
\begin{enumerate}
\item\label{enu:decomposition-edge}
Every edge of $T$ appears in exactly one tree of $\D{T}{L}$.
\item\label{enu:decomposition-size}
The size of every tree in $\D{T}{L}$ is at most $2L+1$ and at least~$2$.
\item\label{enu:decomposition-num}
The number of trees in $\D{T}{L}$ is $O(n/L)$.
\item\label{enu:boundary}
For every $T'\in \D{T}{L}$, at most two nodes of $T'$ can appear
in other trees of $\D{T}{L}$.
These nodes are called the \emph{boundary nodes} of $T'$.
\item\label{enu:boundary2}
A boundary node of a tree $T'\in \D{T}{L}$ can be either a root of $T'$
or a leaf of $T'$.
In the latter case the node will be called the \emph{boundary leaf} of $T'$.
\item\label{enu:intervals}
For every $T'\in \D{T}{L}$, there are at most two maximal
intervals $I_1$ and $I_2$ such that a node $x\in T$ is
a non-root node of $T'$ if and only if the preorder rank of $x$ is in
$I_1 \cup I_2$.
%If the number of these intervals is two, then the boundary leaf is
%the node of preorder rank $\min(\max(I_1),\max(I_2))$.
\end{enumerate}
\end{lemma}

We now describe an algorithm that generates the decomposition of
Lemma~\ref{lem:tree-decomposition}
(the algorithm is based on the algorithm of Farzan and Munro with minor changes).
The algorithm uses a procedure $\pack{x}{x_1,\ldots,x_k}$ that receives
a node $x$ and some children $x_1,\ldots,x_k$ of $x$, where each child $x_i$
has an associated subtree $S_i$ of $T$ that contains $x_i$ and some of its
descendants. Each tree $S_i$ has size at most $L-1$.
The procedure merges the trees $S_1,\ldots,S_k$ into larger trees as follows.
\begin{enumerate}
\item
For each $i$, add the node $x$ to $S_i$, and make $x_i$ the child of $x$.
\item
$i\gets 1$.
\item
Merge the tree $S_i$ with the trees $S_{i+1},S_{i+2},\ldots$
(by merging their roots)
and stop when the merged tree has at least $L$ nodes,
or when there are no more children of $x$.
\item
Let $S_j$ be the last tree merged with $S_i$. If $j<k$, set $i\gets j+1$
and go to step~3.
\end{enumerate}
We say that a node $x\in T$ is \emph{heavy} if $|\subtree{T}{x}| \geq L$,
where $\subtree{T}{x}$ is the subtree of $T$ that contains $x$ and all its
descendants.
A heavy node is \emph{type~2} if it has at least two heavy children,
and otherwise it is \emph{type~1}.

The decomposition algorithm has two phases.
In the first phase the algorithm processes the type~2 heavy nodes.
Let $x$ be a type~2 heavy node and let $x_1,\ldots,x_k$ be children of $x$.
Suppose that the heavy nodes among $x_1,\ldots,x_k$ are
$x_{h_1},\ldots,x_{h_{k'}}$, where $h_1 < \cdots < h_{k'}$.
The algorithm adds to $\D{T}{L}$ the following trees.
\begin{enumerate}
\item A subtree whose nodes are $x$ and the parent of $x$ 
(if $x$ is not the root of $T$).
\item For $i=1,\ldots,k'$, a subtree whose nodes are $x$ and $x_{h_i}$.
\item For $i=1,\ldots,k'+1$, the subtrees generated by
$\pack{x}{x_{h_{i-1}+1},\ldots,x_{h_i-1}}$, where the subtree associated
with each $x_j$ is $\subtree{T}{x_j}$.
We assume here that $h_0 = 1$ and $h_{k'+1} = k+1$.
\end{enumerate}

In the second phase, the algorithm processes maximal paths of type~1 heavy
nodes.
Let $x_1,\ldots,x_k$ be a maximal path of type~1 heavy nodes
($x_i$ is the child of $x_{i-1}$ for all $i$).
If $x_k$ has a heavy child, denote this child by $x'$.
Let $S$ be a subtree of $T$ containing $x_1$ and its descendants, except
$x'$ and its descendants if $x'$ exists.
Let $i$ be the maximal index such that $|\subtree{S}{x_i}| \geq L$.
If no such index exists, $i=1$.
Now, run $\pack{x_i}{y_1,\ldots,y_d}$, where $y_1,\ldots,y_d$ are the
children of $x_i$ in $S$. The subtree associated with each $y_j$ is
$\subtree{S}{y_j}$.
Each tree generated by procedure $\packb$ is added to $\D{T}{L}$.
If $i>1$, add to $\D{T}{L}$ the subtree whose nodes are $\{x_i,x_{i-1}\}$,
and continue recursively on the path $x_1,\ldots,x_{i-1}$.

For a tree $T$ and an integer $L$ we define a tree $\TL{T}{L}$ as follows.
%Construct a tree decomposition $\D{T}{L}$ according to
%Lemma~\ref{lem:tree-decomposition}.
%If the root $r$ of $T$ appears in several trees of $\D{T}{L}$,
%add to $\D{T}{L}$ a tree that consists of $r$.
The tree $\TL{T}{L}$ has a node $v_S$ for every tree $S \in \D{T}{L}$,
and a node $v_r$ which is the root of the tree.
For two trees $S_1,S_2\in\D{T}{L}$, $v_{S_1}$ is the parent of $v_{S_2}$
in $\TL{T}{L}$ if and only if the root of $S_2$ is the boundary leaf of
$S_1$.
The node $v_r$ is the parent of $v_S$ for every $S\in\D{T}{L}$ such that
the root of $S$ is the root of $T$.

%For convenience, we assume that each tree $S\in \D{T}{L}$ has a boundary leaf.
%If the decomposition algorithm does not generate a boundary leaf for $S$
%(in other words, if $v_S$ is a leaf in $\TL{T}{L}$), we arbitrarily pick
%the first leaf of $S$ as the boundary leaf.

\begin{observation}\label{obs:heavy}
For every tree $S\in \D{T}{L}$, if $v_S$ is a leaf of $\TL{T}{L}$,
the only node of $S$ that is a heavy node of $T$ is the root of $S$.
Otherwise, the set of nodes of $S$ that are heavy nodes of $T$
consists of all the nodes on the path
from the root of $S$ to the boundary leaf of $S$.
\end{observation}

\section{Generalized aB-trees}\label{sec:aB-tree}

In this section we describe the aB-tree (augmented B-tree) structure of
Patrascu~\cite{Patrascu08}, and then generalize it.
An \emph{aB-tree} is a data structure that stores an array with elements from a
set $\Sigma$.
Let $\A$ be the set of all such arrays.
Let $B$ be an integer (not necessarily constant),
and let $\func\colon \A \to \Phi$ be a function that has the following property:
There is a function $\funcp\colon \mathbb{N}\times \Phi^B \to \Phi$ such that
for every array $A\in \A$ whose size is dividable by $B$,
$\func(A) = \funcp(|A|,\func(A_1),\ldots,\func(A_B))$, where
$A = A_1 \cdots A_B$ is a partition of $A$ into $B$ equal sized sub-arrays.

Let $A\in \A$ be an array of size $m=B^t$.
An \emph{aB-tree} of $A$ is a $B$-ary tree defined as follows.
The root $r$ of the tree stores $\func(A)$.
The array $A$ is partitioned into $B$ sub-arrays of size $m/B$,
and we recursively build aB-trees for these sub-arrays.
The $B$ roots of these trees are the children of $r$.
The recursion stops when the sub-array has size~1.
%Each node $v$ in the tree corresponds to a sub-array $A_v$ of $A$,
%and the node stores $\func(A_v)$.
%The root corresponds to the whole array $A$, and the array
%$A_{v_1},\ldots,A_{v_B}$ corresponding to the children $v_1,\ldots,v_B$ of a
%node $v$ are obtained by partitioning $A_v$ into $B$ equal sized parts.
%The leaves of the aB-tree correspond to the sub-arrays of $A$ of size~1.

An aB-tree supports queries on $A$ using the following algorithm.
Performs a descent in the aB-tree starting at the root.
At each node $v$, the algorithm decides to which child of $v$ to go by examining
the $\func$ values stored at the children of $v$.
We assume that if these values are packed into one word, the decision is
performed in constant time.
When a leaf is reached, the algorithm returns the answer to the query.
Let $\Num{n}{\alpha}$ denote the number of arrays $A\in \A$ of size $n$
with $\func(A)=\alpha$.
The following theorem is the main result in~\cite{Patrascu08}.
\begin{theorem}\label{thm:aB-tree}
If $B=O(w/\log(|A|+|\Phi|))$ (where $w\geq \log n$ is the word size),
the aB-tree of an array $A$ can be stored using
at most $\log\Num{|A|}{\func(A)}+2$ bits.
The time for performing a query is $O(\log_B |A|)$ using pre-computed tables
of size $O(|\Sigma|+|\Phi|^{B+1}+B\cdot|\Phi|^B)$.
\end{theorem}
We note that the value $\func(A)$ is required in order to answer queries, 
and the space for storing this value is not included in the bound
$\log\Num{|A|}{\func(A)}+2$ of the theorem.

In the rest of this section we generalizes Theorem~\ref{thm:aB-tree}.
Let $\A$ be a set of objects (for example, $\A$ can be a set of
ordered trees).
As before, assume there is a function $\func\colon \A \to \Phi$.
We assume that $\func(A)$ encodes the size of $A$
(namely, $|A|$ can be computed from $f(A)$).
Suppose that there is a \emph{decomposition algorithm} that receives an object
$A\in \A$ and generates sub-objects $A_1,\ldots,A_B$ (some of these objects
can be of size 0)
and a value $\beta \in \Phi_2$ which contains the information necessary to
reconstruct $A$ from $A_1,\ldots,A_B$.
Formally, we denote by $\partition{A} = (\beta,A_1,\ldots,A_B)$ the output of
the decomposition algorithm.
We also define a function $\gfunc\colon \A \to \Phi_2$ by $\gfunc(A) = \beta$
and functions $\func_i\colon \A \to \Phi$ by $\func_i(A) = \func(A_i)$.
%The reconstruction algorithm is a function $\joinb\colon\Phi_2\times \A^B \to \A$
%such that if $\partition{A} = (\beta,A_1,\ldots,A_k)$ then
%$\join{\beta,A_1,\ldots,A_k}=A$.
Let $\F = \{(g(A),f_1(A),\ldots,f_B(A)) \colon A\in \A\}$.
We assume that the decomposition algorithm has the following properties.
\begin{enumerate}
\renewcommand{\labelenumi}{(P\arabic{enumi})}
\item
\label{enu:first}
\label{enu:join-partition}
There is a function $\joinb\colon \Phi_2\times \A^B \to \A$
such that $\join{\partition{A}} = A$ for every $A\in \A$.
% if $\partition{A} = (\beta,A_1,\ldots,A_k)$ then
%$\join{\beta,A_1,\ldots,A_k}=A$.
\item
\label{enu:partition-join}
$\partition{\join{\beta,A_1,\ldots,A_B}} = (\beta,A_1,\ldots,A_B)$
for every $A_1,\ldots,A_B \in \A$ and $\beta\in \Phi_2$ such that
$(\beta,\func(A_1),\ldots,\func(A_B))\in \F$.
\item
\label{enu:fprime}
There is a function $\funcp\colon \F\to \Phi$ such that
$\func(A) = \funcp(\gfunc(A),\func_1(A),\ldots,\func_B(A))$
for every $A\in \A$.
\item
\label{enu:size}
There is a constant $\delta\leq B/2$ such that
if $\partition{A} = (\beta,A_1,\ldots,A_k)$, then
$|A_i| \leq \delta |A|/B$ for all $i$.
\label{enu:last}
\end{enumerate}

Let $\Num{\alpha}{\beta}$ denotes the number of objects $A\in \A$ for which
$\func(A)=\alpha$ and $\gfunc(A)=\beta$.
Let
\[
\mathcal{X}_{\alpha,\beta} = \{(\vec{\alpha},\vec{\beta}) \colon
 \vec{\alpha}=(\alpha_1,\ldots,\alpha_B)\in\Phi^{B},
 \vec{\beta}\in\Phi_2^B,
 (\beta,\alpha_1,\ldots,\alpha_B)\in\mathcal{F},
 \funcp(\beta,\alpha_1,\ldots,\alpha_B)=\alpha
 \}.
\]
\begin{lemma}\label{lem:N-alpha-beta}
For every $\alpha\in\Phi$ and $\beta\in\Phi_2$,
$\sum_{((\alpha_1,\ldots,\alpha_B),(\beta_1,\ldots,\beta_B))\in
\mathcal{X}_{\alpha,\beta}} \prod_{i=1}^B \Num{\alpha_i}{\beta_i} 
= \Num{\alpha}{\beta}$.
\end{lemma}
\begin{proof}
Let $\A_1$ be the set of all tuples $(A_1,\ldots,A_B) \in \A^B$ such that
\[
((\func(A_1),\ldots,\func(A_B)),(\gfunc(A_1),\ldots,\gfunc(A_B)))\in
\mathcal{X}_{\alpha,\beta}.
\]
Let $\A_2$ be the set of all $A\in \A$ such that $\func(A) = \alpha$ and
$\gfunc(A) = \beta$.
We need to show that $|\A_1| = |\A_2|$.
Define a mapping $h$ by
$h(A_1,\ldots,A_B) = \join{\beta,A_1,\ldots,A_B}$.
We will show that $h$ is a bijection from $\A_1$ to $\A_2$.

Fix $(A_1,\ldots,A_B) \in \A_1$ and denote $A=\join{\beta,A_1,\ldots,A_B}$.
By the definition of $\A_1$ and $\mathcal{X}_{\alpha,\beta}$,
$(\beta,\func(A_1),\ldots,\func(A_B))\in\F$,
and by Property~(P\ref{enu:partition-join},
$\partition{A}=(\beta,A_1,\ldots,A_B)$.
Hence, $\func_i(A)=\func(A_i)$ for all $i$ and $\gfunc(A) = \beta$.
We have
$\func(A)=\funcp(\gfunc(A),\func_1(A),\ldots,\func_B(A)) = 
\funcp(\beta,\func(A_1),\ldots,\func(A_B))=\alpha$,
where the first equality follows from Property~(P\ref{enu:fprime})
and the third equality follows from the definition of
$\mathcal{X}_{\alpha,\beta}$.
We also shown above that $\gfunc(A) = \beta$.
Therefore, $h(A_1,\ldots,A_B) \in \A_2$.

The mapping $h$ is injective due to Property~(P\ref{enu:partition-join}).
We next show that $h$ is surjective.
Fix $A\in \A_2$. By definition, $\func(A) = \alpha$ and $\gfunc(A) = \beta$.
Let $\partition{A} = (\beta,A_1,\ldots,A_B)$.
By Property~(P\ref{enu:join-partition}), $h(A_1,\ldots,A_B) = A$, so it remains
to show that $(A_1,\ldots,A_B) \in \A_1$.
By definition,
$(\beta,\func(A_1),\ldots,\func(A_B)) = (\beta,\func_1(A),\ldots,\func_B(A))
\in \F$.
By Property~(P\ref{enu:fprime}),
$\funcp(\beta,\func(A_1),\ldots,\func(A_B))=\func(A)=\alpha$.
Therefore, $(A_1,\ldots,A_B) \in \A_1$.
\end{proof}

We now define a generalization of an aB-tree.
A generalized aB-tree of an object $A\in \A$ is defined as follows.
The root $r$ of the tree stores $\func(A)$ and $\gfunc(A)$.
Suppose that $\partition{A} = (\beta,A_1,\ldots,A_B)$.
Recursively build aB-trees for $A_1,\ldots,A_B$,
and the roots of these trees are the children of $r$.
The recursion stops when the object has size~1 or~0.

The following theorem generalizes Theorem~\ref{thm:aB-tree}.
The proof of the theorem is very similar to the proof of
Theorem~\ref{thm:aB-tree} and uses Lemma~\ref{lem:N-alpha-beta} in order to
bound the space.
\begin{theorem}\label{thm:aB-tree-2}
If $B=O(w/\log(|\Phi|+|\Phi_2|))$,
the generalized aB-tree of an object $A\in \A$ can be stored using at most
$\log\Num{\func(A)}{\gfunc(A)}+2$ bits.
The time for performing a query is $O(\log_B |A|)$ using pre-computed tables
of size $O(a_1+|\Phi|^B\cdot|\Phi_2|^B\cdot(|\Phi|\cdot|\Phi_2|+B))$,
where $a_1$ is the number of objects in $\A$ of size $1$.
\end{theorem}

\section{The data structure}\label{sec:data-structure}
For a tree $T$ with degree distribution $\vec{n} = (n_0,\ldots,n_{n-1})$ define
the tree degree entropy
$\entropy{T} = \frac{1}{n}\sum_{i\colon n_i>0} n_i \log \frac{n}{n_i}$.
Since $n\entropy{T} = \log \mathcal{N}(\vec{n}) + O(\log n)$,
it suffices to show a data structure for $T$ that uses
$n\entropy{T}+O(n/\log^t n)$ bits for any constant $t$.

%Let $T$ be an ordered tree with $n$ nodes,
Let $t$ be some constant.
Define $B=\log^{1/3} n$ and $L=\log^{t+2}n$.
As in~\cite{Patrascu08}, define $\er{i}$ to be the rounding up of
$\log \frac{n}{n_i}$ to a multiple of $1/L$.
%where $n_i$ is the number of nodes in $T$ with $i$ children.
If $n_i = 0$, $\er{i}$ is the rounding up of
$\log n$ to a multiple of $1/L$.
For a tree $S$ define
$\entropyr{S} = \sum_{i=1}^{|S|} \er{\degree{\preorderselectb{S}{i}}}$,
where $\preorderselectb{S}{i}$ is the $i$-th node of $S$ in preorder.
%Note that $\entropyr{T} \leq n\entropy{T}+n/L$.
%The following lemma was stated in~\cite{Patrascu08} without a proof.
%For completeness, we give a proof of the lemma.
Let $\Sigma = \{i\leq n-1\colon n_i > 0 \}$.
We say that a tree $S$ is a \emph{$\Sigma$-tree} if for every node $x$ of $S$,
except perhaps the root, $\degree{x} \in \Sigma$.
\begin{lemma}\label{lem:entropy}
For every $m\leq n$ and $a\geq 0$,
the number of $\Sigma$-trees $S$ with $m$ nodes and $\entropyr{S} = a$
is at most $2^{a+1}$.
\end{lemma}
\begin{proof}
For a string $A$ over the alphabet $\Sigma$ define
$\entropyr{A} = \sum_{i=1}^{|A|} \er{A[i]}$.
Let $\Num{n}{a}$ be the number of strings over $\Sigma$ with length $m$ and
$\entropyr{A} = a$.
We first prove that $\Num{m}{a} \leq 2^a$ using induction on $m$
(we note that this inequality was stated in~\cite{Patrascu08} without a proof).
The base $m=0$ is trivial.
We now prove the induction step.
Let $A$ be a string of length $m$ with $\entropyr{A} = a$.
Clearly, $\er{A[1]} \leq a$ otherwise $\entropyr{A} > a$,
contradicting the assumption that $\entropyr{A} = a$.
If we remove $A[1]$ from $A$, we obtain a string $A'$ of length $m-1$ and
$\entropyr{A'} = \entropyr{A}-\er{\degree{A[1]}}\geq 0$.
Therefore,
$\Num{m}{a} = \sum_{i\in \Sigma\colon \er{i} \leq a} \Num{m-1}{a-\er{i}}$.
Using the induction hypothesis, we obtain that
\[\Num{m}{a} \leq \sum_{i\in\Sigma\colon \er{i} \leq a} 2^{a-\er{i}}
 \leq \sum_{i\in\Sigma} 2^{a-\log \frac{n}{n_i}}
 = 2^a  \sum_{i\in\Sigma} \frac{n_i}{n} = 2^a.
\]

We now bound the number of $\Sigma$-trees with $m$ nodes and $\entropyr{S} = a$.
We say that a $\Sigma$-tree is of \emph{type~1} if the degree of its root is in
$\Sigma$, and otherwise the tree is of \emph{type~2}.
For every $\Sigma$-tree $S$ we associate a string $A_S$ in which
$A_S[i] = \degree{\preorderselectb{S}{i}}$.
If $S$ is a type~1 $\Sigma$-tree then
$A_S$ is a string over the alphabet $\Sigma$ and
$\entropyr{A_S} = \entropyr{S}$.
Therefore, the number of type~1 $\Sigma$-trees $S$ with $m$ nodes and
$\entropyr{S} = a$ is at most $\Num{m}{a} \leq 2^a$.
If $S$ is a type~2 $\Sigma$-tree then
$\substr{A_S}{2}{m}$ is a string over the alphabet $\Sigma$ and
$\entropyr{\substr{A_S}{2}{m}} = \entropyr{S}-a'$, where
$a'$ is the rounding up of $\log n$ to a multiple of $1/L$.
Since there are at most $m$ ways to choose the degree of the root of $S$,
it follows that the number of type~2 $\Sigma$-trees $S$ with $m$ nodes and
$\entropyr{S} = a$ is at most $m\Num{m}{a-a'} \leq n 2^{a-a'} \leq 2^a$.
\end{proof}

%Our data structure is similar the the structure of Farzan and Munro
%(see Section~\ref{sec:Farzan-Munro}).
To build our data structure on $T$,
we first partition $T$ into \emph{macro trees} using the decomposition
algorithm of Lemma~\ref{lem:tree-decomposition} with parameter $L$.
On each macro tree $S$ we build a generalized aB-tree as follows.

%We apply the decomposition algorithm of Lemma~\ref{lem:tree-decomposition}
%on $S$ with parameter $L(S)=\Theta(|S|/B)$, where the constant hidden in the
%$\Theta$ notation is chosen such that the number of trees in the decomposition
%is at most $B$
%(such a constant exists to due to part~\ref{enu:decomposition-num} of
%Lemma~\ref{lem:tree-decomposition}).
%The decomposition algorithm generates subtrees $S_1,\ldots,S_k$ of $S$,
%with $k\leq B$.
%The subtrees $S_1,\ldots,S_k$ are numbered according to the preorder ranks
%of their roots, and two subtrees with a common root are numbered according to
%the preorder rank of the first child of the root.
%If $k<B$ we add empty subtrees $S_{k+1},\ldots,S_B$.
%The root of the aB-tree of $S$ has $B$ children that are built recursively over
%the trees $S_1,\ldots,S_B$.

Let $\A$ be the set of all $\Sigma$-trees with at most $2L+1$ nodes,
and in which one of the leaves may be designated a boundary leaf.
We first describe procedure $\partitionb$.
For a tree $S\in \A$, $\partition{S}$ generates subtrees $S_1,\ldots,S_B$ 
of $S$ by applying the algorithm of Lemma~\ref{lem:tree-decomposition}
on $S$ with parameter $L(S)=\Theta(|S|/B)$, where the constant hidden in the
$\Theta$ notation is chosen such that the number of trees in the decomposition
is at most $B$
(such a constant exists to due to part~\ref{enu:decomposition-num} of
Lemma~\ref{lem:tree-decomposition}).
This algorithm generates subtrees $S_1,\ldots,S_k$ of $S$, with $k\leq B$.
The subtrees $S_1,\ldots,S_k$ are numbered according to the preorder ranks
of their roots, and two subtrees with a common root are numbered according to
the preorder rank of the first child of the root.
If $k<B$ we add empty subtrees $S_{k+1},\ldots,S_B$.

We next describe the mappings $\func\colon \A \to \Phi$ and
$\gfunc\colon \A \to \Phi_2$.
Recall that $\gfunc(S)$ is the information required to reconstruct $S$ from
$S_1,\ldots,S_B$. In our case, $\gfunc(S)$ is the balanced parenthesis
string of the tree $\TL{S}{L(S)}$.
The number of nodes in $\TL{S}{L(S)}$ is $k+1$.
Since $k\leq B$, $\gfunc(S)$ is a binary string of length at most $2B+2$.
Thus, $|\Phi_2|=O(2^{2B})$. % $=2^{O(\sqrt{n})}$.

We define $\func(S)$ to be a vector
$(\entropyr{S},\fsize{S},\fsizeright{S},\fsizeboundary{S},\fsizemax{S},
\fdegree{S},\fdistance{S},\fpreorder{S})$ whose components are
defined as follows.
\begin{itemize}
\item
$\fsizeright{S}=|\subtree{S}{x}|$, where $x$ is the rightmost child of
the root of $S$ (recall that $\subtree{S}{x}$ is the subtree of $S$ containing
$x$ and its descendants).
\item
$\fsizeboundary{S}=|\subtree{S}{x'}|$, where $x'$ is the child of the root of
$S$ which is on the path between the root of $S$ and the boundary leaf of $S$.
If $S$ does not have a boundary leaf, $\fsizeboundary{S}=0$.
\item
$\fsizemax{S}=\max_y |\subtree{S}{y}|$ where the maximum is taken over
every node $y$ of $S$ whose parent is on the path between the root of $S$
and the boundary leaf of $S$, and $y$ is not on this path.
If $S$ does not have a boundary leaf,
%$\fsizemax{S}=\max_y |\subtree{S}{y}|$, where
the maximum is taken over all children $y$ of the root of $S$.
\item
$\fdegree{S}$ is the number of children of the root of $S$.
\item
$\fdistance{S}$ is the distance between the root of $S$ and the boundary leaf
of $S$.
If $S$ does not have a boundary leaf, $\fdistance{S} = 0$.
\item
$\fpreorder{S}$ is the number of nodes in $S$ that appear before
the boundary leaf of $S$ in the preorder of $S$.
If $S$ does not have a boundary leaf, $\fpreorder{S} = 0$.
\end{itemize}
We note that the value $\entropyr{S}$ is required in order to bound the space
of the aB-trees.
The values $\fsize{S}$, $\fsizeright{S}$, $\fsizeboundary{S}$, $\fsizemax{S}$,
$\fdegree{S}$, and $\fdistance{S}$ are required in order to satisfy
Property~(P\ref{enu:partition-join}) of Section~\ref{sec:aB-tree}
(see the proof of Lemma~\ref{lem:properties-2} below).
These values are also used for answering queries.
Finally, the value $\fpreorder{S}$ is needed to answer queries.

The values $\fsize{S},\fsizeright{S},\fsizeboundary{S},\fsizemax{S},\fdegree{S},
\fdistance{S},\fpreorder{S}$ are integers bounded by $L$.
Moreover, $\entropyr{S}$ is a multiple of $1/L$ and
$\entropyr{S} \leq L(\log n+1/L)=L\log n+1$.
Therefore, $|\Phi| = O(L^7 \cdot L^2\log n) = O(L^9 \log n)$.
It follows that the condition $B=O(w/\log(|\Phi|+|\Phi_2|))$ of
Theorem~\ref{thm:aB-tree-2} is satisfied
(since $B=\log^{1/3} n$ and
$w/\log(|\Phi|+|\Phi_2|) = \Omega(w/B) = \Omega(\log^{2/3}n)$).
Moreover, the size of the lookup tables of Theorem~\ref{thm:aB-tree-2}
is $O(2^{2B(B+1)})=O(\sqrt{n})$.

The following lemmas shows that Properties~(P\ref{enu:first})--(P\ref{enu:last})
of Section~\ref{sec:aB-tree} are satisfied.
\begin{lemma}\label{lem:properties-1}
Property~(P\ref{enu:join-partition}) is satisfied.
\end{lemma}
\begin{proof}
We define the function $\joinb$ as follows.
Given a balanced parenthesis string $\beta$ of a tree $S_\beta$
and trees $S_1,\ldots,S_B$, the tree $S = \join{\beta,S_1,\ldots,S_B}$
is constructed as follows.
For $i=1,\ldots,B$, associate the tree $S_i$ to the node
$\preorderselectb{S_\beta}{i+1}$.
For every internal node $v$ in $S_\beta$,
merge the boundary leaf of the tree $S_i$ associated with $v$,
and the roots of the trees associated with the children
of $v$
(if $v$ is the root of $S_\beta$ just merge the roots of the trees associated
with the children of $v$).
By definition, $\join{\partition{S}} = S$ for every tree $S$.
\end{proof}

\begin{lemma}\label{lem:boundary-size-degree}
Let $S=\join{\beta,S_1,\ldots,S_B}$.
If a node $x\in S$ is a boundary node of some tree $S_i$,
the values of $|\subtree{S}{x}|$ and $\degree{x}$ can be computed from
$\beta,\func(S_1),\ldots,\func(S_B)$.
\end{lemma}
\begin{proof}
Let $S_\beta$ be the tree whose balanced parenthesis string is $\beta$.
Assume that $x$ is not the root of $S$ (the proof for the case when $x$ is the
root is similar).
Therefore, $x$ is the boundary leaf of some tree $S_i$.
Let $I$ be a set containing every index $j\neq i$ such that
$\preorderselectb{S_\beta}{j+1}$ is a descendant of
$\preorderselectb{S_\beta}{i+1}$.
Observe that $|\subtree{S}{x}| = 1+\sum_{j\in I} (\fsize{S_j}-1)$.
Similarly, let $I_2$ be a set containing every index $j$ such that
$\preorderselectb{S_\beta}{j+1}$ is a child of
$\preorderselectb{S_\beta}{i+1}$.
By part~\ref{enu:decomposition-edge} of Lemma~\ref{lem:tree-decomposition},
$\degree{x} = \sum_{j\in I_2} \fdegree{S_j}$.
The lemma now follows since $I,I_2$ can be computed from $\beta$ and
$\fsize{S_j},\fdegree{S_j}$ are components of $\func(S_j)$.
\end{proof}

\begin{lemma}\label{lem:properties-2}
Property~(P\ref{enu:partition-join}) is satisfied.
\end{lemma}
\begin{proof}
Suppose that $\beta\in \Phi_2$ is a balanced parenthesis string
and $S_1,\ldots,S_B \in \A$ are trees such that
$(\beta,\func(S_1),\ldots,\func(S_B))\in \F$
(recall that $\F = \{(g(S),f_1(S),\ldots,f_B(S)) \colon S\in \A\}$).
We need to show that
$\partition{\join{\beta,S_1,\ldots,S_B}} = (\beta,S_1,\ldots,S_B)$.
Denote $S = \join{\beta,S_1,\ldots,S_B}$.
By the definition of $\F$, there is a tree $S^*$ such that
$\gfunc(S^*) = \beta$ and $\func_i(S^*) = \func(S_i)$ for all $i$.
Denote $\partition{S^*} = (\beta,S^*_1,\ldots,S^*_B)$.
Let $S_\beta$ be a tree whose balanced parenthesis string is $\beta$.

By Lemma~\ref{lem:boundary-size-degree}, $|S|=|S^*|$ and therefore
$L(S) = L(S^*)$.
Recall that a node of $S$ or $S^*$ is heavy if the size of its subtree is at
least $L(S)$.
Define the \emph{skeleton} of a tree to be the subtree that contains
the heavy nodes of the tree.
We first claim that the skeleton of $S^*$ can be reconstructed from
$\beta,\func(S^*_1),\ldots,\func(S^*_B)$.
To prove this claim,
define trees $P_1,\ldots,P_B$, where $P_i$ is a path of length
$\fdistance{S^*_i}$.
By Observation~\ref{obs:heavy},
the skeleton of $S^*$ is isomorphic to $\join{\beta,P_1,\ldots,P_B}$.

We now show that $S$ and $S^*$ have isomorphic skeletons.
Consider some subtree $S_i$ such that $\preorderselectb{S_\beta}{i+1}$ is not
a leaf of $S_\beta$.
Let $x$ be the boundary leaf of $S_i$,
and let $x^*$ be the boundary leaf of $S^*_i$.
By Lemma~\ref{lem:boundary-size-degree} and Observation~\ref{obs:heavy},
$|\subtree{S}{x}| = |\subtree{S^*}{x^*}| \geq L(S^*) = L(S)$, so
$x$ is a heavy node of $S$.
Therefore, all the nodes of $S$ that are on the path
between the root of $S_i$ and the boundary leaf of $S_i$ are heavy nodes of $S$
(this follows from the fact that all ancestors of a heavy node are heavy).
Let $S'$ be the subtree of $S$ containing all the nodes of $S$ that are
nodes on the path between the root and the boundary leaf of $S_i$,
for every $S_i$ such that $\preorderselectb{S_\beta}{i+1}$ is not a leaf of
$S_\beta$.
Since $\fdistance{S_i} = \fdistance{S^*_i}$ for all $i$, it follows that
$S'$ is isomorphic to $\join{\beta,P_1,\ldots,P_B}$ and to the
skeleton of $S^*$.
It remains to show that $S'$ is the skeleton of $S$.
Assume conversely that there is a heavy node $y$ of $S$ which is not in $S'$.
We can choose such $y$ whose parent $x$ is in $S'$.
Let $S_i$ be the tree containing $y$.
% (note that $y$ is not the root of $S_i$).
Since the $y$ is not on the path between the root and the boundary leaf of
$S_i$,
all the descendants of $y$ are in $S_i$.
Since $x$ is on the path between the root and the boundary leaf of $S_i$
(if $S_i$ does not have a boundary leaf, $x$ is the root of $S_i$),
$\fsizemax{S_i} \geq |\subtree{S}{y}| \geq L(S) = L(S^*)$.
It follows that $\fsizemax{S^*_i} = \fsizemax{S_i} \geq L(S^*)$
which means that $S^*_i$ has a heavy node which is not on the path between
the root and the boundary leaf.
This contradicts Observation~\ref{obs:heavy}.
Therefore, $S$ and $S^*$ have isomorphic skeletons.

We now prove that $\partition{S} = (\beta,S_1,\ldots,S_B)$.
Suppose we run the decomposition algorithm on $S$ and on $S^*$.
In the first phase of the algorithm, the algorithm processes type~2 heavy
nodes.
Since $S$ and $S^*$ have isomorphic skeletons, there is a bijection between
the type~2 heavy nodes of $S$ and the type~2 heavy nodes of $S^*$.
Let $x$ be a type~2 heavy node of $S$ and let $x^*$ be the corresponding
type~2 heavy node of $S^*$.
%By the decomposition algorithm, $x^*$ is the boundary leaf of a subtree
%$S^*_i$ that consists of $x^*$ and its parent.
Let $x^*_1,\ldots,x^*_k$ be the children of $x^*$, and let
$x^*_{h_1},\ldots,x^*_{h_{k'}}$ be the heavy children of $x^*$.
When processing $x^*$, the decomposition algorithm generates the following
subtrees of $S^*$.
\begin{enumerate}
\item A subtree whose nodes are $x^*$ and its parent.
\item For $j=1,\ldots,k'$, a subtree whose nodes are $x^*$ and $x^*_{h_j}$.
\item For $j=1,\ldots,k'+1$, the subtrees generated by
$\pack{x^*}{x^*_{h_{j-1}+1},\ldots,x^*_{h_j-1}}$.
\end{enumerate}
For every subtree $S^*_j$ of the first two types above that is generated when
processing $x^*$, the subtree $S_j$ is generated when processing $x$
(since the number of heavy children of $x$ is equal to the number of heavy
children of $x^*$).
We now consider the subtrees of the third type.
Suppose without loss of generality that $h_1>1$. Consider the call to
$\pack{x^*}{x^*_1,\ldots,x^*_{h_1-1}}$.
% and suppose that the first tree generated by this call is $S^*_a$.
The first tree generated by this call, denoted $S^*_{a}$, consists of $x^*$,
some children $x^*_1,\ldots,x^*_l$ of $x^*$,
and all the descendants of $x^*_1,\ldots,x^*_l$, where
$l = \fdegree{S^*_{a}}$.
From the definition of procedure $\packb$,
$\sum_{j=1}^{l-1}|\subtree{S^*}{x^*_j}| < L(S^*)-1$.
Additionally, if $l < h_1-1$,
$\sum_{j=1}^{l}|\subtree{S^*}{x^*_j}| \geq L(S^*)-1$.

%Let $S^*_{a'}$ be the subtree of $S^*$ containing $x^*$ and $x^*_{h_1}$.
%Let $I$ be a set containing every index $j<a'$ such that
%$\preorderselectb{S_\beta}{j+1}$ is a child of
%$\preorderselectb{S_\beta}{i+1}$.
Let $I = \{\preorderrank{x^*_j}-1\colon j=1,\ldots,h_1-1\}$.
We have that $h_1-1 = \sum_{j\in I} \fdegree{S^*_j}$.
The number of children of $x$ before the first heavy child of $x$ is equal to
$\sum_{j\in I} \fdegree{S_j} = \sum_{j\in I} \fdegree{S^*_j} = h_1-1$.
Let $x_1,\ldots,x_{h-1}$ be these children.

Since $\fdegree{S_{a}} = \fdegree{S^*_{a}} = l$,
when the decomposition algorithm processes the node $x$ of $S$ we have
\[
\sum_{j=1}^{l-1}|\subtree{S}{x_j}|
= \fsize{S_{a}}-\fsizeright{S_{a}} - 1
= \fsize{S^*_{a}}-\fsizeright{S^*_{a}} - 1
= \sum_{j=1}^{l-1}|\subtree{S^*}{x^*_j}|
< L(S^*)-1 = L(S)-1.
\]
Additionally, if $l < h_1-1$,
\[
\sum_{j=1}^{l}|\subtree{S}{x_j}|
= \fsize{S_{a}} - 1
= \fsize{S^*_{a}} - 1
= \sum_{j=1}^{l}|\subtree{S^*}{x^*_j}|
\geq L(S^*) = L(S).
\]
Therefore, the first tree generated by $\pack{x}{x_1,\ldots,x_{h_1-1}}$
is $S_{a}$. Continuing with the same arguments, we obtain that
for every tree $S^*_j$ generated by a call to $\pack{x^*}{\cdot}$ when
processing $x^*$,
the tree $S_j$ is generated by a call to $\pack{x}{\cdot}$ when
processing $x$.
%if the children in $S_\beta$ of $\preorderselectb{S_\beta}{a'}$ have preorder
%ranks $a_1 < a_2 < \cdots$, 
%then the trees $S_{a_1},S_{a_2},\ldots$ are generated when processing $x$.

Now consider the second phase of the algorithm.
Let $x^*_1,\ldots,x^*_k$ be a maximal path of type~1 heavy nodes of $S^*$,
and let $x_1,\ldots,x_k$ be the corresponding maximal path of type~1 heavy
nodes of $S$.
For simplicity, assume that $x^*_k$ does not have a heavy child.
%Let $i$ be the maximal index such that $|\subtree{S^*}{x^*_i}| \geq L(S^*)$,
%or $i=1$ if no such index exists.
Let $S^*_{a_1},S^*_{a_2},\ldots$
be the subtrees generated by $\pack{x^*_i}{y^*_1,y^*_2,\ldots}$, where
$y^*_1,y^*_2,\ldots$ are the children of $x^*_i$.
Let $S^*_a$ be the subtree from $S^*_{a_1},S^*_{a_2},\ldots$ that contains
$x^*_k$.
Let $l = \fdistance{S^*_a} = \fdistance{S_a}$.
By the definition of the decomposition algorithm,
$\fsizeboundary{S^*_a} = |\subtree{S^*}{x^*_{k-l+1}}| < L(S^*)$.
Moreover, if $l < k-1$,
$1+\sum_j (\fsize{S^*_{a_j}}-1) = |\subtree{S^*}{x^*_{k-l}}| \geq L(S^*)$.
Therefore,
$|\subtree{S}{x_{k-l+1}}| = \fsizeboundary{S_a} = \fsizeboundary{S^*_a} < L(S)$
and  if $l < k-1$,
$|\subtree{S}{x_{k-l}}| = 1+\sum_j (\fsize{S_{a_j}}-1)
= 1+\sum_j (\fsize{S^*_{a_j}}-1) \geq L(S)$.
Therefore, when processing the path $x_1,\ldots,x_k$, the decomposition
algorithm makes a call to $\pack{x_i}{y_1,y_2,\ldots}$, where
$y_1,y_2,\ldots$ are the children of $x_i$.
Using the same argument used for the first phase of the algorithm, we obtain
that the trees $S_{a_1},S_{a_2},\ldots$ are generated by
$\pack{x_i}{y_1,y_2,\ldots}$.
\end{proof}

\begin{lemma}\label{lem:properties-3}
Property~(P\ref{enu:fprime}) is satisfied.
\end{lemma}
\begin{proof}
Let $S$ be a tree and $\partition{S} = (\beta,S_1,\ldots,S_B)$.
Recall that $\func(S)=(\entropyr{S},\fsize{S},\fsizeright{S},\fsizeboundary{S},
\fsizemax{S},\fdegree{S},\fdistance{S},\fpreorder{S})$
and $\gfunc(S) = \beta$ is the balanced parenthesis string of $\TL{S}{L(S)}$.
A node $x$ of $S$ is called an \emph{inner boundary node} if it is a boundary
node of some subtree $S_i$.

By definition,
$\entropyr{S}$ is equal to
$\sum_{i=1}^B (\entropyr{S_i}-\er{\fdegree{S_i}})$ plus the sum of
$\er{\degree{x}}$ for every inner boundary node $x$ of $S$.
By Lemma~\ref{lem:boundary-size-degree}, every such value $\er{\degree{x}}$
can be computed from $\gfunc(S),\func(S_1),\ldots,\func(S_B)$.
Therefore, $\entropyr{S}$ can be computed from
$\gfunc(S),\func(S_1),\ldots,\func(S_B)$.

Similarly, $\fsize{S}$ is equal to $\sum_{i=1}^B (\fsize{S_i}-1)$ plus
the number of inner boundary nodes of $S$.
The number of inner boundary nodes of $S$ is equal to the number of internal nodes in
$\TL{S}{L(S)}$.
Thus, $\fsize{S}$ can be computed from $\gfunc(S),\func(S_1),\ldots,\func(S_B)$.

We next consider $\fsizeright{S}$.
Let $x$ be the rightmost child of the root of $S$.
Let $v_{S_i}$ be the rightmost child of the root of $\TL{S}{L(S)}$.
Then, the tree $S_i$ contains both the root of $S$ and $x$.
If $x$ is not the boundary leaf of $S_i$ then all the descendants of $x$ are
in $S_i$. Thus, $\fsizeright{S} = \fsizeright{S_i}$.
Otherwise, by Lemma~\ref{lem:boundary-size-degree},
$\fsizeright{S}$ can be computed from $\gfunc(S),\func(S_1),\ldots,\func(S_B)$.

The other components of $\func(S)$ can also be computed from
$\gfunc(S),\func(S_1),\ldots,\func(S_B)$. We omit the details.
\end{proof}

\begin{lemma}\label{lem:properties-4}
Property~(P\ref{enu:size}) is satisfied.
\end{lemma}
\begin{proof}
The lemma follows from part~\ref{enu:decomposition-size} of
Lemma~\ref{lem:tree-decomposition}.
\end{proof}
Our data structure for the tree $T$ consists of the following components.
\begin{itemize}
\item For each macro tree $S$, the aB-tree of $S$, stored
using Theorem~\ref{thm:aB-tree-2}.
\item For each macro tree $S$, the values $\func(S)$ and $\gfunc(S)$.
\item Additional information and data structures for handling queries, which
will be described later.
\end{itemize}

The space of the aB-trees and the values $\func(S),\gfunc(S)$ is
$\sum_S (\log \Num{\func(S)}{\gfunc(S)}+2+\lceil \log |\Phi|\rceil
+ \lceil \log |\Phi_2|\rceil )$, where the summation is over every
macro tree $S$.
By part~\ref{enu:intervals} of Lemma~\ref{lem:tree-decomposition},
\[
\sum_S (2+\lceil \log |\Phi|\rceil + \lceil \log |\Phi_2|\rceil )
= O(n/L\cdot B) = O(n/\log^t n).
\]
We next bound $\sum_S \log \Num{\func(S)}{\gfunc(S)}$.
Since $\entropyr{S},|S|$ are components of $\func(S)$, we have from
Lemma~\ref{lem:entropy}
that $\Num{\func(S)}{\gfunc(S)} \leq 2^{\entropyr{S}+1}$.
Therefore, $\sum_S \log \Num{\func(S)}{\gfunc(S)} \leq \sum_S (\entropyr{S}+1)$.
By definition,
$\sum_S \entropyr{S}$ is equal to $\entropyr{T}+\sum_S \er{d_S}$
minus the sum of $\er{\degree{x}}$ for every node $x$ of $T$ which is a boundary
node of some macro tree.
Therefore,
\[\sum_S \entropyr{S} \leq \entropyr{T}+\sum_S \er{\fdegree{S}}
\leq (n\entropy{T}+O(n/L))+O(n/L \cdot \log n)
= n\entropy{T}+O(n/\log^t n).
\]

Most of the queries on $T$ are handled in a similar way these queries are
handled in the data structure of Farzan and Munro~\cite{FarzanM14}.
We give some examples below.
We assume that a node $x$ in $T$ is represented by its preorder number.
In order to compute the macro tree that contains a node $x$, we store
the following structures. %(as in the structure of Farzan and Munro):
\begin{itemize}
\item A rank-select structure on a binary string $B$ of length $n$ in which
$B[x] = 1$ if nodes $x$ and $x-1$ belong to different macro trees.
\item An array $M$ in which $M[i]$ is the number of the macro tree that contains
node $x=\select{1}{B}{i}$. % (the index of the $i$-th one in $B$).
\end{itemize}
By part~\ref{enu:intervals} of Lemma~\ref{lem:tree-decomposition},
the number of ones in $B$ is $O(n/L)$.
Therefore, the space for $B$ is
$O(n/L\cdot\log L)+O(n/\log^t n) = O(n/\log^t n)$ bits
(using the rank-select structure of Patrascu~\cite{Patrascu08}),
and the space for $M$ is $O(n/L\cdot \log n) = O(n/\log^t n)$ bits.

For handling $\depth{x}$ queries, the data structure stores
the depths of the roots of the macro trees. The required space is
$O(n/L\cdot \log n)=O(n/\log^t n)$ bits.
Answering a $\depth{x}$ query is done by finding the macro tree $S$ containing
$x$.
Then, add the depth of the root of $S$ (which is stored in the data structure)
to the distance between $x$ and the root of $S$.
The latter value is computed using the aB-tree of $S$.
It suffices to describe how to compute this value when the aB-tree is
stored naively.
Recall that the root of the aB-tree corresponds to $S$,
and the children of the root corresponds to subtrees $S_1,\ldots,S_B$ of $S$.
Finding the subtree $S_i$ that contains $x$ can be done using a lookup table
indexed by $\gfunc(S)$, $\fsize{S_1},\ldots,\fsize{S_B}$,
and $\fpreorder{S_1},\ldots,\fpreorder{S_B}$.
Next, compute the distance between the root of $S_i$ and the root of $S$
using a lookup table indexed by $\gfunc(S)$ and
$\fdistance{S_1},\ldots,\fdistance{S_B}$.
Then the query algorithm descend to the $i$-th child of the root of the aB-tree
and continues the computation in a similar manner.

The handling of level ancestor queries is different than the way these queries
are handled in the structure of Farzan and Munro.
%In these queries, $\levelancestor{x}{d}$ is the ancestor $y$ of $x$ whose
%distance from $x$ is $d$.
We define weights on the edges of $\TL{T}{L}$ as follows.
For every non-root node $v_S$ in $\TL{T}{L}$, the weight of the edge between
$v_S$ and its parent is $\fdistance{S}$.
The data structure stores a weighted ancestor structure on $\TL{T}{L}$.
We use the structure of Navarro and Sadakane~\cite{NavarroS14} which has
$O(1)$ query time.
The space of this structure is
$O(n'\log n'\cdot\log (n'W)+n'W/\log^{t'}(n'W))$ for every constant $t'$, where
$n' = |\TL{T}{L}|$ and $W$ is the maximum weight of an edge of $\TL{T}{L}$.
Since $n' = O(n/L)$ and $W = O(L)$, we obtain that the space is $O(n/\log^t n)$
bits.

In order to answer a $\levelancestor{x}{d}$ query, first find the macro
tree $S$ that contains $x$.
Then use the aB-tree of $S$ to find $\levelancestor{x}{d}$ if this node is
in $S$. Otherwise, let $r$ be the root of $S$ and let $d'$ be the distance
between $r$ and $x$ ($d'$ is computed using the aB-tree).
Next, perform a $\levelancestor{\parent{v_S}}{d-d'}$ on $\TL{T}{L}$, and let
$v_{S'}$ be the answer.
Let $v_{S''}$ be the child of $v_{S'}$ which is an ancestor of $v_S$.
The node $\levelancestor{x}{d}$ is in the macro tree $S''$,
and it can be found using a query on the aB-tree of $S''$.

\bibliographystyle{plain}
\bibliography{ds,dekel}
\end{document}